\documentclass[journal]{IEEEtran}

\ifCLASSINFOpdf

\else

\fi
\usepackage[english]{babel}
\usepackage[utf8x]{inputenc}
\usepackage{amsmath}
\usepackage{amsthm}
\usepackage{graphicx}
\usepackage[colorinlistoftodos]{todonotes}
\usepackage{mathtools}
\usepackage{subcaption}
\usepackage{amssymb}
\usepackage[export]{adjustbox}
\usepackage{stfloats}
\usepackage{lettrine}
\usepackage{tikz}
\usetikzlibrary{calc}

\newtheorem{theorem}{Theorem}

\begin{document}
\title{Federated Learning in MIMO Satellite Broadcast System}

\author{Raphael Pinard,~\IEEEmembership{Student Member,~IEEE,}, Mitra Hassani,
        Wayne Lemieux,~\IEEEmembership{Fellow,~IEEE,}}

%

\markboth{}%
{Shell \MakeLowercase{\textit{et al.}}: Bare Demo of IEEEtran.cls for IEEE Journals}

\maketitle

\begin{abstract}
Federated learning (FL) is a type of distributed machine learning at the wireless edge that preserves the privacy of clients' data from adversaries and even the central server. Existing federated learning approaches either use (i) secure multiparty computation (SMC) which is vulnerable to inference or (ii) differential privacy which may decrease the test accuracy given a large number of parties with relatively small amounts of data each. To tackle the problem with the existing methods in the literature, In this paper, we introduce incorporate federated learning in the inner-working of MIMO systems.
\end{abstract}


\IEEEpeerreviewmaketitle

\section{Introduction}
\lettrine{C}{hannel} capacity is defined as the maximum rate at which data can be transmitted at an arbitrary small error probability \cite{gallager1968information}.
The capacity of a single-input–single-output (SISO) additive white Gaussian noise (AWGN) channel was first addressed by Shannon \cite{shannon1948mathematical}. With the invention of powerful space-time coding scheme \cite{tarokh1998space}, \cite{foschini1996layered}, the application of multiple-inputs multiple-output (MIMO) systems has been unraveled to  both academia and industry as mean of achieving channel capacity far beyond that of traditional techniques. This novel approach offers a unique
solution for increasing demand for high performance next generation wireless communications.
\\
In this paper, MIMO satellite communication (SatCom) is considered where the receiver is a linear array antenna. We aim to find the channel capacity for low Earth orbiting satellites whose positions are unknown to the terrestrial receiver. In low orbiting satellites, line of sight (LOS) becomes more dominant and the path loss reduces. Although its efficiency in a rich scattering environment has already been
demonstrated \cite{corazza2007digital}, less is studied under this new scenario. 
\\
The typical way to analyze channel capacity of MIMO systems is to find eigenvalue distribution for the channel matrix multiplied by its conjugate \cite{liang2005ergodic}. In our setup, when the channel is modeled as pure LoS, the channel matrix $H$ would become a Vandermonde matrix \cite{horn2012matrix}. However, for this setup, eigenvalue distribution is unknown and only there have been some studies over asymptotic behavior of it \cite{tucci2011eigenvalue,hamidi2022over,hamidi2019systems}. In \cite{tucci2011eigenvalue,struhsaker2020methods}, the author found a lower and upper bound for the maximum eigenvalue and presented the channel capacity for a sufficiently large matrices. 
In this paper, the average channel capacity and outage probability of such channels are analyzed, assuming the receiver has the perfect channel state information (CSI). This has been done by approximating the channel capacity and the accuracy of our method is justified with simulations.  
\\
The paper is organized as follows. We present the MIMO system model in Section \ref{Model}. Some basics to find the channel capacity of a MIMO system is reviewed in section \ref{sec:MIMO}. in section \ref{E(C)} the average channel capacity of MIMO SatCom systems, by finding moments of channel matrix multiplied by its conjugate, is presented. In \ref{sec:outage}, we discuss the outage capacity of MIMO SatCom systems. In section \ref{sec:max}, the optimum satellite arrangement is found in terms of maximum channel capacity.

\subsection{Definitions and Assumptions}
\begin{itemize}
\item{In this paper we use $x^{\star}$, $x^{\prime}$ and $x^{\dag}$ to show conjugate, transpose and conjugate transpose of a vector x, respectively. Also, $n_{R}$ and $n_{T}$ are the number of receivers and transmitters respectively. }
\item{A complex random vector $z$ of size $n\time 1$  is said to be Gaussian if the $2n\times 1$ real vector $\hat{x}$ consisting of its real and imaginary parts, i.e. \begin{equation*}
\hat{x} =
\begin{bmatrix} 
\Re(x) \\
\Im(x)
\end{bmatrix},
\end{equation*}
is Gaussian.}
\item{The expectation and covariance of $\hat{x}$ are defined as $\mathbf{E}[\hat{x}] \in \Re^{2n}$ and $\mathbf{E}[(\hat{x}-\mathbf{E(\hat{x})})(\hat{x}-\mathbf{E(\hat{x})})^\dag] \in \Re^{2n}$, respectively.}
\item{A complex Gaussian random vector x is circularly symmetric ~\cite{telatar1999capacity} if the covariance of the corresponding $\hat{x}$ has the structure
\begin{equation*}
\mathbf{E}[(\hat{x}-\mathbf{E(\hat{x})})(\hat{x}-\mathbf{E(\hat{x})})^\dag]=\frac{1}{2}
\begin{bmatrix} 
\Re(Q) & -Im(Q)\\
\Im(Q) & Re(Q)
\end{bmatrix},
\end{equation*}
for some Hermitian non-negative definite matrix $Q\in \mathbb{C}^{n\times n}$}
\end{itemize}

\section{MIMO System Model} \label{Model}
Let us consider a MIMO system in which the transmitters comprise $n_{T}$ satellites revolving around the Earth that their positions are unknown for the user located on the Earth. We further assume that all these satellites are in visibility range of the server and located in a spherical cap. This cap could be defined as $\phi \in (-\pi, \pi)$ and $\theta \in (0 ,\frac{\pi}{2})$, where $\phi$ and $\theta$ are azimuth and elevation angles in standard spherical coordinate. Also, the receiver is a linear antenna made of $n_{R}$ antenna elements.
\\
The transmitted signals in each symbol period are denoted by a $t\times 1$ complex vector x, where the $x_{i}$ refers to the transmitted signal from antenna $i$. The total power of the complex transmitted signal $x$ is constrained to $P$ regardless of the number of transmit antennas, meaning that: 
\begin{align}
\mathbf{E}[x^{\dag}x]=trace(\mathbf{E}[xx^{\dag}])=P 
\end{align} 
The transmitted signal bandwidth is narrow enough, so its
frequency response can be considered as flat. The received
signal $y$ is given by
\begin{align}
y=\mathbf{H}x+n
\end{align}
Where H is a $n_{R}\times n_{T}$ complex channel gain matrix. The $ij^{th}$ entry of the matrix $H$ represents the channel gain from the $j^{th}$ transmit to the $i^{th}$ receive antenna.
The noise at the receiver is denoted by the $r\times 1$ vector $n$. We make an assumption that components of $n$ are statistically independent complex zero mean Gaussian random variables with independent and equal variance of real and imaginary parts. The covariance matrix of n is given by
\begin{align}
\mathbf{E}[nn^{\dag}]=\sigma^2\mathbf{I}_r
\end{align}
Where $\sigma^2$ is the identical noise power at each of the receive antennas.
We assume that the total power per receive antenna is equal to the total transmitted power. In other worlds, signal attenuations and amplifications in the propagation process are ignored. We
also assume that H is perfectly known to the receiver, but not at the transmitter. The channel matrix can be estimate at
the receiver by transmitting a training sequence. The estimated
CSI can be communicated to the transmitter via a reliable
feedback channel.
\section{MIMO channel Capacity} \label{sec:MIMO}
For an evenly spaced antenna array along y-axis, array factor (AF) is (Appendix \ref{app:AF})
\begin{align} \label{eq:4}
AF(\theta, \phi)=\sum_{m=0}^{M-1} I_m e^{jkmdsin(\theta)sin(\phi)}
\end{align}
where $I_m$ is the gain injected to the $m^{th}$ antenna, $k=\frac{2\pi}{\lambda}$ and $\theta$ and $\phi$ are elevation and azimuth angles in standard spherical coordinate. In our setup we assume that all $I_m$'s$~=1$.
\\
On the other hand, each column of matrix H corresponds to the gain received from one satellite over the antenna array, and therefore in each column of H, there is only two random variables ($\theta$ and $\phi$) and the other entries are determined based on the AF formula. Mathematically speaking, H could be written as:
\\

$\mathbf{H}=\frac{1}{\sqrt{n_Tn_R}}\times
\\
\begin{pmatrix*}[c]
1 &...& 1 \\
e^{jkd\sin(\theta_1)\sin(\phi_1)} &. . .&  e^{jkd\sin(\theta_{n_T})\sin(\phi_{n_T})}\\
e^{jk2d\sin(\theta_1)\sin(\phi_1)} & . . .&  e^{jk2d\sin(\theta_{n_T})\sin(\phi_{n_T})}\\
.& ...& .\\
.& ...& .\\
.& ...& .\\
e^{jk (n_{R}-1)d\sin(\theta_1)\sin(\phi_1)} & . . .& e^{jk (n_{R}-1)d\sin(\theta_{n_T})\sin(\phi_{n_T})}\\
\end{pmatrix*}
$
\\
\\
where $\theta_i \sim U(0,\frac{\pi}{2})$ and $\phi_i \sim U(-\pi,\pi)$ to define a symmetric spherical cap above the terrestrial server. This configuration is depicted in Figure \ref{fig1}.
\\

For the case of study, where H is assumed to be known at the receiver and not the transmitter from~\cite{telatar1995capacity} we have: 
\begin{align}
C=\max\limits_{{p(x):tr(\mathbf{E}[xx^{\dag}])=P}} I (x;y,\mathbf{H})
\end{align}

The power received from each satellite is assumed to be equal to $P$. In~\cite{liang2005ergodic} the author simplified the above equation to obtain

\begin{align} \label{eq:6}
C=\mathbf{E}_H \Big\{  \log_2 \Big( det(I_{n_R} +\frac{P}{\sigma^2}HH^\dag \Big) \Big\}
\end{align}

\begin{align}  \label{eq:7}
=\mathbf{E}_H \Big\{  \log_2 \Big( det(I_{n_T} +\frac{P}{\sigma^2}H^\dag H \Big) \Big\}
\end{align}

Equation \eqref{eq:6} is used when $n_R<n_T$ and  \eqref{eq:7} is for when $n_T \leq n_R$. Also we define:  

\[ \mathbf{W}=
    \begin{dcases}  
	 H^\dag H & n_T\leq n_R 
	\\ HH^\dag & n_R<n_T
    \end{dcases}
\]
To find the expected capacity based on equation \eqref{eq:7}, it is common to write $HH^\dag$ in terms of its eigenvalues. This means that~\cite{liang2005ergodic}
\begin{align}  \label{eq:8}
C=\mathbf{E}_{\lambda} \Big\{\sum_{i=1}^{n_T}  \log_2 \big( 1+\frac{P}{\sigma^2} \lambda_i \big)\Big\}
\end{align}
Where $\lambda$'s are the eigenvalues of matrix $ \mathbf{W}$.
\\
On the other hand we know that $\lambda$ depends on $\theta$ and $\phi$ and therefor equation \eqref{eq:8} becomes:
\begin{align}  \label{eq:9}
C=\sum_{i=1}^{n_T} \mathbf{E} \Big\{  \log_2 { \big( 1+ \frac{P}{\sigma^2 } \lambda_i (\theta, \phi) \big)}  \Big\}
\end{align}
or equivalently 
\begin{align}  \label{eq:10}
C=\sum_{i=1}^{n_T} \int  \int \ \log_2 \big( 1+ \frac{P}{\sigma^2} \lambda_i \big) f(\theta,\phi) d\theta d\phi 
\end{align}
where $ f(\theta,\phi)$ is the joint probability distribution between random variables $\theta$ and $\phi$.

\section{Finding the average channel capacity} \label{E(C)}
It could be simply proven (see appendix \ref{th:lambda}) that the trace of a matrix is equal to the summation of its eigenvalues. Therefore, for random $N\times N $matrix $\mathbf{A}$ defined on probability space of $\zeta$ one can say
\begin{align}  \label{eq:11}
\mathbf{E} \Big\{Trace(\mathbf{A}) \Big\}=\mathbf{E}\sum_{i=1}^{N} \lambda_i=\sum_{i=1}^{N}\mathbf{E}\Big\{ \lambda_i \Big\}=\sum_{i=1}^{N} \int_{\zeta} \lambda_i f(\theta,\phi) d\theta d\phi
\end{align}
\begin{theorem} \label{th:lambda}
The trace of a matrix to its kth power is equal to the sum of its eigenvalues raised to the kth power.
\end{theorem}
\begin{proof}
See Appendix \ref{th:lambda}.
\end{proof}
Thus, similarly we have
\begin{align}  \label{eq:12}
\mathbf{E} \Big\{Trace(\mathbf{A^k}) \Big\}=\sum_{i=1}^{N} \int_{\zeta} \lambda_i^k f(\theta,\phi) d\theta d\phi
\end{align}

On the other hand, we can approximate the $\log$ term in equation \ref{eq:10} by using Taylor expansion \cite{hummel1949generalization} as follows
\begin{align}  \label{eq:13}
\log_2 \big( 1+ \frac{P}{\sigma^2} \lambda_i \big)=\frac{1}{\ln(2)} \sum_{k=1}^{\infty} \frac{(-1)^{k+1}  (\frac{P}{\sigma^2}\lambda_i)^{k}}{k}
\end{align}
Plugging this expansion is in equation \eqref{eq:10} yields to
\begin{align}  \label{eq:14}
C=\frac{1}{\ln(2)} \sum_{k=1}^{\infty}  \frac{(-1)^{k+1} P^k}{k~\sigma^{2k}} \sum_{i=1}^{n_T} \int \int \lambda_i^k f(\theta,\phi) d\theta d\phi 
\end{align}
Because $\mathbf{W}$ is a $N\times N$ matrix whose eigenvalues are $\lambda$'s, by using equation \eqref{eq:12} in the above equation we have

\begin{align}  \label{eq:15}
C=\frac{1}{\ln(2)} \sum_{k=1}^{\infty}  \frac{(-1)^{k+1} P^k}{k~\sigma^{2k}} \mathbf{E} \Big\{Trace(\mathbf{W^k}) \Big\}
\end{align}
Thus, one can find the average capacity for the channel defined in this setup by finding $ \mathbf{E} \Big\{Trace(\mathbf{W^k}) \Big\}$. 
\\
In finding average capacity by \eqref{eq:15}, infinite terms must be evaluated among which there is no relation. However, by simulation we show that only three terms of this sigma would give us the average capacity with a high accuracy. 
In the following parts, it is attempted to find the first three terms of the above summation, which indeed requires to evaluate $ \mathbf{E} \Big\{Trace(\mathbf{W^k}) \Big\}$ for $k=1,2,3$.
\\
Later in this section, we justify our approximation. 
\subsection{ Evaluating  the First Term, $ \mathbf{E} \Big\{Trace(\mathbf{W}) \Big\}$}
We denote the entry of matrix $\mathbf{W}$ located in $i^{th}$ row and $j^{th}$ column by $w_{ij}$. Also, we define 
\begin{align}  \label{eq:16}
\gamma_{ij}= \sin(\theta_j)\sin(\phi_j)-\sin(\theta_i)\sin(\phi_i)
\end{align}
By the definition of $\mathbf{W}$ we therefore have
\begin{align}  \label{eq:17}
w_{ij}=\frac{1}{n_Tn_R}\sum_{m=0}^{n_R-1} e^{jkmd\gamma_{ij}}
\end{align}
It is clear that $w_{ii}=\frac{n_R}{n_Rn_T}=\frac{1}{n_T}$ for all $i=1,2,...,n_T$ and so for the first moment of trace we have:
\begin{align}  \label{eq:18}
\mathbf{E} \Big\{Trace(\mathbf{W}) \Big\}=Trace(\mathbf{W})=n_T\times \frac{1}{n_T}=1
\end{align}

\subsection{ Evaluating  the Second Term, $ \mathbf{E} \Big\{Trace(\mathbf{W}^2) \Big\}$}
First, we write the trace of $\mathbf{W}^2$ in terms of its entries 
\begin{align}  \label{eq:19}
Trace(\mathbf{W}^2)=\sum_{i=1}^{N_T} \sum_{j=1}^{N_T} w_{ij}w_{ji}=\sum_{i=1}^{N_T} \sum_{j=1}^{N_T} w_{ij}w_{ij}^{\star} \nonumber \\
=n_R^2n_T+2\sum_{i=1}^{n_T} \sum_{i>j}^{n_T} w_{ij}w_{ij}^{\star}~~~~~~~~~~~~
\end{align}
Where the second last equality is due to the fact that $\mathbf{W}$ is Hermitian. To put away the cases in which $i=j$ -this would be useful in following calculation-, the last equation in \eqref{eq:19} was derived.

To better analyze the summation, we expand the $ w_{ij}w_{ij}^{\star}$ term as follows

\begin{equation*}
\begin{multlined}
w_{ij}w_{ij}^{\star}=\frac{1}{(n_Rn_T)^2}\Big(1+e^{-jkd\gamma_{ij}}+...+e^{-jk(n_R-1)d\gamma_{ij}}+
\\
...+e^{jkd\gamma_{ij}}+1+...+e^{jk(n_R-2)d\gamma_{ij}}+
\\
...+e^{jk(n_R-1)d\gamma_{ij}}+e^{jk(n_R-2)d\gamma_{ij}}+...+1\Big)
 \end{multlined}
\end{equation*}
It is seen that every term would be added to its conjugate, meaning that it could be simplified to $\cos$ terms. By putting away all $1$'s the following equation is derived
\begin{align}  \label{eq:20}
w_{ij}w_{ij}^{\star}=\frac{1}{(n_Rn_T)^2}\Big(n_R+2\sum_{s=1}^{n_R-1} (n_R-s)\cos(skd\gamma_{ij})\Big)
\end{align}
Now by plugging equation \eqref{eq:20} into \eqref{eq:19} we obtain
\begin{align}  \label{eq:21}
(n_Rn_T)^2 \times Trace(\mathbf{W}^2)= \nonumber ~~~~~~~~~~~~~~~~~~~~~~~~~~~~~~~~~~~~~\\
n_R^2n_T+2\sum_{i=1}^{n_T} \sum_{i>j}^{n_T} \Big( n_R+2\sum_{s=1}^{n_R-1} (n_R-s)\cos(skd\gamma_{ij}) \Big)&&
\\*
=n_R^2n_T+n_Rn_T(n_T-1)+\nonumber ~~~~~~~~~~~~~\\
4\sum_{i=1}^{n_T} \sum_{i>j}^{n_T}\sum_{s=1}^{n_R-1} (n_R-s)\cos(skd\gamma_{ij})~~~~~~~~~
\end{align}
Considering that $\gamma_{ij}$ is a function of $(\theta_i, \phi_i,\theta_j,\phi_j)$, to find the expected value of $Trace(\mathbf{W}^2)$ we have:
\begin{align}  \label{eq:23}
(n_Rn_T)^2 \times \mathbf{E} \Big\{Trace(\mathbf{W}^2) \Big\}=n_R^2n_T+n_Rn_T(n_T-1)+\nonumber ~~~~~~~~~~~~~\\
4\sum_{i=1}^{n_T} \sum_{i>j}^{n_T}\sum_{s=1}^{n_R-1} (n_R-s) \int_\zeta \cos(skd\gamma_{ij})~~~~~~~~~~~~~
\end{align}
where $\zeta$ is the probability space.
\\
Henceforth, we get down to finding $\int_\zeta \cos(skd\gamma_{ij})$ by substituting back the original value for $\gamma_{ij}$. Also,  
\begin{align}  \label{eq:24}
\cos(skd( \sin(\theta_j)\sin(\phi_j)-\sin(\theta_i)\sin(\phi_i)))= \nonumber \\
\cos(skd(\sin(\theta_j)\sin(\phi_j)))\cos(skd(\sin(\theta_i)\sin(\phi_i)))- \nonumber \\
\sin(skd(\sin(\theta_j)\sin(\phi_j)))\sin(skd(\sin(\theta_i)\sin(\phi_i)))
\end{align}

Meanwhile, we recall that $\theta_i$ and $\theta_j$  $\sim U(0,\frac{\pi}{2})$, and $\phi_i$ and $\phi_j$  $\sim U(-\pi,\pi)$, all of which are independent random variables; thus

\begin{align}  \label{eq:25}
\int_\zeta \cos(skd\gamma_{ij})=~~~~~~~~~~~~~~~~~~~~~~~~~~~~~~~~ \nonumber \\
(\frac{1}{\frac{\pi}{2}}\frac{1}{2\pi})^2\Big(\int_{0}^{\frac{\pi}{2}}\int_{-\pi}^{\pi} \cos(skd(\sin(\theta_j)\sin(\phi_j))) d(\phi_j)d(\theta_j) \times  ~~~~~~ \nonumber \\ 
\int_{0}^{\frac{\pi}{2}}\int_{-\pi}^{\pi}\cos(skd(\sin(\theta_i)\sin(\phi_i)))d(\phi_i)d(\theta_i))\Big) - ~~~~~~~~~~\nonumber \\ 
(\frac{1}{\frac{\pi}{2}}\frac{1}{2\pi})^2\Big(\int_{0}^{\frac{\pi}{2}}\int_{-\pi}^{\pi} \sin(skd(\sin(\theta_j)\sin(\phi_j))))d(\phi_j)d(\theta_j)\times  ~~~~~~\nonumber \\ 
\int_{0}^{\frac{\pi}{2}}\int_{-\pi}^{\pi}\sin(skd(\sin(\theta_i)\sin(\phi_i)))) d(\phi_i)d(\theta_i) \Big) = ~~~~~~~~~ \\ \label{eq:26}
(\frac{1}{\frac{\pi}{2}}\frac{1}{2\pi})^2\Big(\int_{0}^{\frac{\pi}{2}}\int_{-\pi}^{\pi} \cos(skd(\sin(\theta_j)\sin(\phi_j))))d(\phi_j)d(\theta_j)\times  ~~~~~~ \nonumber \\ 
\int_{0}^{\frac{\pi}{2}}\int_{-\pi}^{\pi}\cos(skd(\sin(\theta_i)\sin(\phi_i)))) d(\phi_i)d(\theta_i)\Big)= ~~~~~~~~~\\
(\frac{1}{\frac{\pi}{2}})^2\Big(\int_{0}^{\frac{\pi}{2}}J_0(skd\sin(\theta_j))d(\theta_j) \times \int_{0}^{\frac{\pi}{2}}J_0(skd\sin(\theta_i))d(\theta_i) \Big)~~~~ \label{eq:27}
\end{align}
Where equation \eqref{eq:26} is hold because $\sin$ is an odd function and the inner integral is symmetric from $-\pi$ to $\pi$. 
\\
In equation \eqref{eq:27} $J_0(.)$ is the Bessel function of the first kind. Also because $\theta_i$ and $\theta_j$ have the same distribution, we can further simplify it to
\begin{align}  \label{eq:28}
\int_\zeta \cos(skd\gamma_{ij})=(\frac{1}{\frac{\pi}{2}})^2 \Big(\int_{0}^{\frac{\pi}{2}}J_0(skd\sin(\theta_j))d(\theta_j)\Big)^2 \nonumber \\
=(\frac{1}{\frac{\pi}{2}})^2 \Big( \int_{0}^{1} \frac{J_0(skd\mu)}{\sqrt{1-\mu^2}}d\mu\Big)^2~~~~~~~~~~~~~
\end{align}
Where the last equation is found by putting $\sin(\theta_j)=\mu$.
From \cite{prudnikov1986integrals}, part 2.12.21, we can find the above integral as follows:
\begin{align}  \label{eq:29}
=(\frac{1}{\frac{\pi}{2}})^2 \Big(\int_{0}^{1} \frac{J_0(skd\mu)}{\sqrt{1-\mu^2}}d\mu \Big)^2 \nonumber \\ 
(\frac{1}{\frac{\pi}{2}})^2 \Big(\frac{\pi}{2} J_0  \big( \frac{skd}{2} \big)^2\Big)^2  \nonumber \\ 
=J_0  \big( \frac{skd}{2} \big)^4
\end{align}
Therefore, the equation \eqref{eq:23} becomes
\begin{align} \label{eq:30}
\mathbf{E} \Big\{Trace(\mathbf{W}^2) \Big\}=\frac{1}{(n_Rn_T)^2}\Big(n_R^2n_T+n_Rn_T(n_T-1)+\nonumber ~~~~~~~~~~~~~\\
2\sum_{s=1}^{n_R-1}(n_R-s)n_T(n_T-1) J_0  \big( \frac{skd}{2} \big)^4\Big)~~~~~~~~~~~~~
\end{align}

\subsection{ Evaluating  the Third Term, $\mathbf{E} \Big\{Trace(\mathbf{W}^3) \Big\} $}
The formulas written in this part, because of mathematical complexity, are briefly mentioned without proof. One can simply prove them all by taking the same steps as the previous part.
\\
The formula to find $ Trace(\mathbf{W}^3) $ is as follows
\begin{align} \label{eq:31}
Trace(\mathbf{W}^3)=\frac{n_R^3n_T}{(n_Rn_T)^3}+ 6\sum_{i=1}^{n_T} \sum_{j>i}^{n_T}w_{ij} \sum_{s=i}^{n_T}w_{si}w_{js}+ \nonumber \\
\sum_{i=1}^{n_T} \sum_{j\neq i}^{n_T}w_{ij} \sum_{s\neq i,j}^{n_T}w_{si}w_{js}~~~~~~~~~~~~~~
\end{align}
and thus, by finding the expected value of the above equation with the same approach as the expected value for $Trace(\mathbf{W}^2)$, one can prove that

\begin{flalign}
\mathbf{E} \Big\{Trace(\mathbf{W}^3) \Big\}= \frac{1}{(n_Rn_T)^3} \Bigg(n_R^3nT+3n_R^2n_T(n_T-1)~~~~~~ \nonumber \\
+n_Rn_T(n_T-1)(n_T-2)~~~~~~ ~~~~~~ \nonumber \\
+6\sum_{s=1}^{n_R-1}n_R(n_R-s)n_T(n_T-1) J_{0} \big( \frac{skd}{2} \big)^2~~~~~~ ~~~~~~ \nonumber \\+6\sum_{s=1}^{n_R-1}(n_T-2)\prod_{k=0}^2(k+s)J_{0} \big( \frac{(n_R-s)kd}{2} \big)^2~~~~~ ~~~\nonumber \\
+6\sum_{s=1}^{\frac{n_R-1}{2}}(n_T-2)\prod_{k=0}^2(n_R-2s+k)J_{0} \big( \frac{skd}{2} \big)^2 J_{0} \big( {skd} \big)~~~~~ ~~~\nonumber \\
+\Big(6\sum_{s=1}^{n_R-1}\sum_{t=1}^{n_R-2s} (n_R-2m)^{+}(n_T-2) \times~~~~~ ~~~ \nonumber \\
\prod_{k=0}^2 (2n_R-2t-2s+2+k) \times ~~~~~ ~~~\nonumber \\
J_{0} \big( \frac{skd}{2} \big) J_{0} \big( \frac{(s+t)kd}{2} \big)J_{0} \big( \frac{(2s+t)kd}{2} \big) \Big)\Bigg)~~~~ ~~~~~
\end{flalign}

In the following figure, the channel capacity for a setup with 64 receive antenna and 64 satellite is simulated. The average capacity is found by generating 100000 random matrices. The results show the accuracy of first-three-term approximation.
\begin{figure}[htbp]
\centering{\includegraphics[scale=0.5]{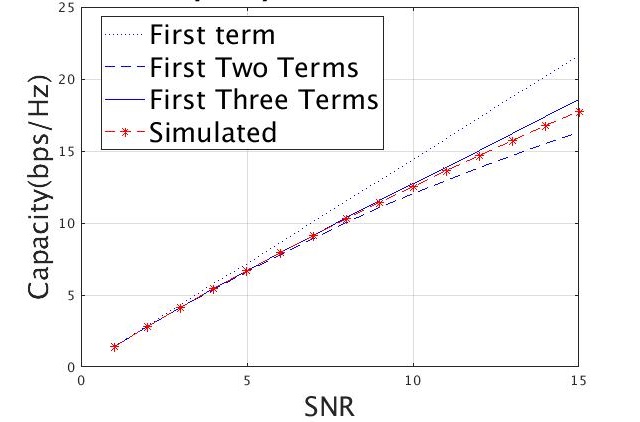}}
    \caption{Channel Capacity for $64\times 64$ MIMO channel.}
    \label{CDD}
\end{figure}

\bibliographystyle{IEEEtran}
\bibliography{refs}
\end{document}